\documentclass{article}

\usepackage{cite}
\usepackage{amsmath,amssymb,amsfonts}
\usepackage{algorithmic}
\usepackage{graphicx, subcaption}
\usepackage{textcomp}
\usepackage{pgf}
\usepackage{bm}

\def\BibTeX{{\rm B\kern-.05em{\sc i\kern-.025em b}\kern-.08em
		T\kern-.1667em\lower.7ex\hbox{E}\kern-.125emX}}

\usepackage{amsthm}
\newtheorem{definition}{Definition}
\newtheorem{problem}{Problem}
\newtheorem{lemma}{Lemma}
\newtheorem{theorem}{Theorem}
\newtheorem{corollary}{Corollary}

\DeclareMathOperator*{\trace}{tr}
\DeclareMathOperator*{\minimize}{minimize}
\DeclareMathOperator*{\subject}{subject\, to}
\DeclareUnicodeCharacter{2212}{-}

\providecommand{\keywords}[1]
{
	\small	
	\textbf{\textit{Keywords---}} #1
}

\begin{document}
	
	\title{Fusion of Distance Measurements between Agents with Unknown Correlations}
	\author{Colin Cros, Pierre-Olivier Amblard,\\
		Christophe Prieur and Jean-François Da Rocha
		\thanks{C. Cros, P.-O. Amblard and C. Prieur are with the CNRS, Univ. Grenoble Alpes, GIPSA-lab, F-38000 Grenoble, France.
			{\tt\small colin.cros@gipsa-lab.fr,
				christophe.prieur@gipsa-lab.fr,
				pierre-olivier.amblard@cnrs.fr}.}
		\thanks{C. Cros and J.-F. Da Rocha are with Telespazio FRANCE, F-31100 Toulouse, France.
			{\tt\small jeanfrancois.darocha@telespazio.com}.}
	}
	\date{}	
	
	\maketitle
	
	\begin{abstract}
		Cooperative localization is a promising solution to improve the accuracy and overcome the shortcomings of GNSS. Cooperation is often achieved by measuring the distance between users. To optimally integrate a distance measurement between two users into a navigation filter, the correlation between the errors of their estimates must be known. Unfortunately, in large scale networks the agents cannot compute these correlations and must use consistent filters. A consistent filter provides an upper bound on the covariance of the error of the estimator taking into account all the possible correlations. In this paper, a consistent linear filter for integrating a distance measurement is derived using Split Covariance Intersection. Its analysis shows that a distance measurement between two agents can only benefit one of them, {i.e.}, only one of the two can use the distance measurement to improve its estimator. Furthermore, in some cases, none can. A necessary condition for an agent to benefit from the measurement is given for a general class of objective functions. When the objective function is the trace or the determinant, necessary and sufficient conditions are given.
	\end{abstract}
	
	
	\section{Introduction}
	
	{Accurate} positioning is a key challenge for numerous strategic applications.  Global Navigation Satellite Systems  (GNSSs) provide a low-cost and effective solution for achieving satisfactory accuracy in open-sky environments. However, in GNSS-denied environments such as urban canyons, indoors or underwater, other methods must be used. Cooperative localization is a promising alternative because it can slow down the loss of accuracy in the absence of GNSS signal \cite{roumeliotis2003analysis}. A simple form of cooperation is the measurement of distances between the users, hereafter called \emph{agents}. Inter-agent distances can be inexpensively determined by measuring the Received Signal Strength or the Time-of-Flight of a signal. For example, cooperative localization has been used for terrestrial vehicles \cite{hossain2016cooperative} and UAVs \cite{chen2020distributed, causa2021improving}. Recent research has demonstrated the effectiveness of cooperation in significantly improving positioning accuracy while using simple algorithms \cite{minetto2019trade}.
	
	One of the challenges in cooperative systems when integrating inter-agent distance measurements is dealing with the correlations between the errors of the agents' estimators. The use of inter-agent measurements tends to correlate these errors. If the cross-covariances are not taken into account, they lead to an underestimation of the errors and a potential divergence of the filter \cite{arambel2001covariance}. These problems can be avoided by calculating and storing the cross-covariances between each pair of estimators. In centralized systems, the estimators of all agents are stacked into a global estimator whose covariance includes all the cross-covariances. This global estimator is then updated in a centralized manner (at a computing station), or in a distributed manner \cite{roumeliotis2002distributed}. However, such solutions are difficult to implement for networks with a large number of agents due to the computational and transmission costs. Another solution to avoid correlation problems involves ensuring that no redundant information is used. In \cite{bahr2009consistent}, the authors propose to keep a table of all past interactions to ensure that only independent information is used in the filter. This requires storing the history of past interactions and may be prohibitive for large networks.
	
	In order to use two-way cooperation when some correlation information is unknown, \emph{covariance consistent} filters should be used. The literature also refers to these as ``\emph{conservative} filters'', see, {e.g.}, \cite{forsling2022conservative}. \emph{Covariance consistency} ensures that the estimator is not overconfident by providing an upperbound on the estimator error. Several methods have been developed to \emph{fuse} different estimators with unknown covariances. A recent paper \cite{forsling2022conservative} unifies some of them in a general framework called the Conservative Linear Unbiased Estimator (CLUE).
	Among these methods,  Covariance Intersection (CI) \cite{julier1997nondivergent} and its extension the Split Covariance Intersection (SCI) \cite{julier2001general} are particularly efficient: several papers, see, {e.g.}, \cite{chen2002estimation}, have described the efficiency of CI, and more recently CI has been shown to be the best method to fuse two estimators under unknown correlations \cite{reinhardt2015minimum}. Even if SCI has not been shown to be optimal, its simplicity makes it a widely used method for decentralized cooperative positioning \cite{li2013cooperative, lima2021data, pierre2018range}. The integration of distance measurements using SCI have already been considered by \cite{julier2007using, noack2011automatic} in the context of SLAM. The authors take benefit from the fact that the observation noise is assumed independent and therefore provides \emph{new} information. However, to the best of our knowledge, no specific study on the \emph{usefulness} of the integration of a distance measurement between two agents with an unknown correlation using SCI has been performed in the literature.
	
	In this paper, we consider cooperative networks where agents estimate their positions and store their own estimators. Here, cooperation occurs through distance measurements between agents. We address the problem of integrating, through linear filtering, a distance measurement between two agents who have \emph{covariance consistent} estimators of their states but the cross-covariance between these estimators is unknown. Throughout the paper, we only consider the integration of a distance measurement. In particular, the dynamics of the system is beyond the scope of this paper and is the subject of further work. Our first contribution is the derivation of the optimal SCI filter for this problem, optimal w.r.t. a given increasing cost function. In practice, this function is often the trace or the determinant of the covariance of the estimate. Our second contribution is the analysis of the SCI filter which shows whether a distance measurement can be used to improve the estimation. We provide a necessary condition in the general case, and necessary and sufficient conditions for the cases where the cost function is the trace or the determinant.
	
	The paper is organized as follows. Section~\ref{sec: Problem statement} introduces the problem of optimal filtering, then Section~\ref{sec: SCI for distance measurement} derives a candidate: the SCI filter. Section~\ref{sec: Usefulness of the filtering} investigates the usefulness of such a filter, and Section~\ref{sec: Discussion} discusses the results, in particular in simulations. Finally, Section~\ref{sec: Conclusion} provides some perspectives and future directions.

	
	\bigbreak
	\textbf{Notation.}
	In the sequel, vectors are underlined {e.g.}, $\underline{x} \in \mathbb{R}^n$, random variables are denoted in lowercase boldface letters {e.g.}, $\bm z$ for a scalar or $\underline{\pmb{x}}$ for a vector, and matrices are denoted in uppercase boldface variables {e.g.}, $\bm{M} \in \mathbb{R}^{n\times n}$. The notation $\mathrm{E}[\cdot]$ denotes the expected value of a random variable and $\left\lVert\cdot\right\rVert$ the Euclidean norm of a vector. The trace, the determinant, the inverse and the transpose of a matrix $\bm{M}$ and the identity matrix are denoted as $\trace \bm{M}$, $\left|\bm{M}\right|$, $\bm{M}^{-1}$, $\bm{M}^\intercal$ and $\bm{I}$ respectively. For two matrices $\bm{A}$ and $\bm{B}$, $\bm{A} \preceq \bm{B}$ means that the difference $\bm{B} - \bm{A}$ is positive semi-definite. A positive definite matrix $\bm{P}$ is represented in the figures by the ellipsoid $\mathcal{E}_{\bm{P}} = \left\{\underline{x} \ | \  \underline{x}^\intercal \bm{P}^{-1}\underline{x} \le 1\right\}$.

	\section{Problem statement}\label{sec: Problem statement}
	
	Let us first recall the definition of a covariance consistent estimator.
	\begin{definition}
		An estimator $(\hat{\underline{\pmb{x}}}, \bm{P})$ of a random variable $\underline{\pmb{x}}$ is said to be \emph{covariance consistent} (or shortly \emph{consistent} in the sequel) if $\mathrm{E}\left[\underline{\pmb{\tilde x}}\right] = 0$ and $\bm{\tilde P} \preceq \bm{P}$ where $\underline{\pmb{\tilde x}} = \underline{\pmb{x}} - \underline{\pmb{\hat x}}$ denotes the error and $\bm{\tilde P} = \mathrm{E}\left[\underline{\pmb{\tilde x}}\underline{\pmb{\tilde x}}^\intercal\right]$ the mean-squared error (MSE).
	\end{definition}
	
	In other words, a consistent estimator is unbiased and does not underestimate the covariance of the error.
	
	Consider two agents, denoted $A$ and $B$, characterized by their states $\underline{\pmb{x}}_A$ and $\underline{\pmb{x}}_B$ in $\mathbb{R}^n$. For the sake of simplicity and without loss of generality, the states are assumed to contain only the positions of the agents; in practice, they may also contain their orientations or their velocities, for example. The two agents have consistent estimators of their states denoted $(\underline{\pmb{\hat x}}_A, \bm{P}_A)$ and $(\underline{\pmb{\hat x}}_B, \bm{P}_B)$.
	Notice that the true (centralized) covariances of the errors, defined as,
	\begin{equation}\label{eq: Centralized covariance}
		\bm{\tilde P} 
		= \mathrm{E}\left[\begin{pmatrix} \underline{\pmb{\tilde x}}_A \\ \underline{\pmb{\tilde x}}_B \end{pmatrix}\begin{pmatrix} \underline{\pmb{\tilde x}}_A \\ \underline{\pmb{\tilde x}}_B \end{pmatrix}^\intercal\right]
		= \begin{bmatrix}
			\bm{\tilde P}_A & \bm{\tilde P}_{AB} \\ \bm{\tilde P}_{AB}^\intercal & \bm{\tilde P}_B
		\end{bmatrix},
	\end{equation} is unknown, {i.e.}, $\bm{\tilde P}_A$, $\bm{\tilde P}_B$, and $\bm{\tilde P}_{AB}$ are unknown. However, consistency restrains the set of possible $\bm{\tilde P}$ to:
	\begin{equation}
		\boldsymbol{\mathcal{P}} = \left\{\begin{bmatrix}
			\bm{\tilde P}_A & \bm{\tilde P}_{AB} \\ \bm{\tilde P}_{AB}^\intercal & \bm{\tilde P}_B
		\end{bmatrix}\succeq \bm{0} \ | \  \bm{\tilde P}_A \preceq \bm{P}_A, \bm{\tilde P}_B \preceq \bm{P}_B\right\}.
	\end{equation}
	
	Furthermore, consider $\bm{z} = \left\lVert\underline{\pmb{x}}_A - \underline{\pmb{x}}_B\right\rVert + \bm{\tilde z}$ a measurement of the distance between $A$ and $B$  where $\bm{\tilde z}$ denotes the error. The error $\bm{\tilde z}$ is assumed centered, $\mathrm{E}\left[\bm{\tilde z}\right]= 0$, with variance $\mathrm{E}\left[\bm{\tilde z}^2\right]=\sigma_m^2$, and independent of the errors of the estimators $\underline{\pmb{\tilde x}}_A$ and $\underline{\pmb{\tilde x}}_B$.
	Finally, consider that the measurement can be linearized around the means of the estimators $\underline{\hat x}_{A}$ and $\underline{\hat x}_{B}$. Introducing the unit-length director vectors  $\underline{u}_{BA} = \frac{\underline{\hat x}_A - \underline{\hat x}_B}{\left\lVert\underline{\hat x}_A - \underline{\hat x}_B\right\rVert}$, the linearized observation writes:
	\begin{equation}
		\bm{z} = \underline{u}_{BA}^\intercal \left(\underline{\pmb{x}}_A - \underline{\pmb{x}}_B\right) + \bm{\tilde z},
	\end{equation}
	where $\bm{\tilde z}$ is still assumed independent from the errors of the estimators. 
	This assumption is necessary to use SCI which is a linear fusion. It is quite idealistic and causes a loss of precision as the distance is a nonlinear function. It is reasonable if the agents are sufficiently far from each other, as the second-order terms become negligible. Without this assumption, no linear filter can be used.
	
	\medbreak
	
	The objective is to \emph{improve} the estimator of $A$ using the estimator of $B$ and the distance measurement to create a \emph{better} estimator $\left(\underline{\pmb{\hat x}}_F, \bm{P}_F\right)$ of the state of $A$. Throughout this paper, the estimators are compared w.r.t. an increasing cost function $J$ (increasing in the sense of the Loewner ordering {i.e.}, $\bm{P} \prec \bm{Q} \implies J(\bm{P}) < J(\bm{Q})$). For two consistent estimators $(\underline{\pmb{\hat x}}_1, \bm{P}_1)$ and $(\underline{\pmb{\hat x}}_2, \bm{P}_2)$ of the same random variable $\underline{\pmb{x}}$, $(\underline{\pmb{\hat x}}_1, \bm{P}_1)$ is said to be \emph{better} than $(\underline{\pmb{\hat x}}_2, \bm{P}_2)$ if $J(\bm{P}_1) < J(\bm{P}_2)$. The estimator $\left(\underline{\pmb{\hat x}}_F, \bm{P}_F\right)$ is designed as an unbiased\footnote{A bias would increase the MSE $\bm{\tilde P}_F$ in the Loewner ordering sense.} \emph{linear filter} defined as follows.
	\begin{definition}
		A \emph{linear filter} for the state of $A$ is an estimator $(\underline{\pmb{\hat x}}_F, \bm{P}_F)$ of $\underline{\pmb{x}}_A$ where $\underline{\pmb{\hat x}}_F$ is a linear combination of $\underline{\pmb{\hat x}}_A$, $\underline{\pmb{\hat x}}_B$ and $\bm{z}$. It is defined by two matrices $\bm{K}_F, \bm{L}_F$ and a vector $\underline{w}_F$ such that:
		\begin{equation}
			\underline{\pmb{\hat x}}_F = \bm{K}_F \underline{\pmb{\hat x}}_A + \bm{L}_F \underline{\pmb{\hat x}}_B + \underline{w}_F \bm{z}.
		\end{equation}
	\end{definition}
	
	To have $\underline{\pmb{\hat x}}_F$ unbiased, since $\mathrm{E}[\bm{z}] = \underline{u}_{BA}^\intercal \left(\mathrm{E}\left[\underline{\pmb{x}}_A\right] - \mathrm{E}\left[\underline{\pmb{x}}_B\right]\right)$ the gains must be dependent and satisfy:
	\begin{align}\label{eq: Relation gains}
		\bm{K}_F &= \bm{I} - \underline{w}_F \underline{u}_{BA}^\intercal ,&
		\bm{L}_F &= \underline{w}_F \underline{u}_{BA}^\intercal.
	\end{align}
	An unbiased linear filter is then only defined by $\underline{w}_F$ as:
	\begin{equation}\label{eq: Linear filter}
		\underline{\pmb{\hat x}}_F = \left(\bm{I} - \underline{w}_F\underline{u}_{BA}^\intercal\right) \underline{\pmb{\hat x}}_A + \underline{w}_F\underline{u}_{BA}^\intercal \underline{\pmb{\hat x}}_B + \underline{w}_F \bm{z}.
	\end{equation}
	For a possible $\bm{\tilde P} \in \boldsymbol{\mathcal{P}}$, the covariance of the error $\underline{\pmb{\tilde x}}_F =  \underline{\pmb{x}}_A - \underline{\pmb{\hat x}}_F$ is:
	\begin{multline}\label{eq: MSE}
		\bm{\tilde P}_F = \bm{\tilde P}_A + \left(\tilde \sigma_A^2 + \tilde \sigma_B^2 -2 \tilde \gamma + \sigma_m^2\right) \underline{w}_F \underline{w}_F^\intercal \\
		- \left(\bm{\tilde P}_A - \bm{\tilde P}_{AB}\right)\underline{u}_{BA}\underline{w}_F^\intercal  - \underline{w}_F\underline{u}_{BA}^\intercal \left(\bm{\tilde P}_A - \bm{\tilde P}_{AB}^\intercal\right)
	\end{multline}
	where
	\begin{align}
		\tilde\sigma_A^2 &= \underline{u}_{BA}^\intercal \bm{\tilde P}_A \underline{u}_{BA}, & \tilde\sigma_B^2 &= \underline{u}_{BA}^\intercal \bm{\tilde P}_B \underline{u}_{BA},\\
		\tilde \gamma &= \underline{u}_{BA}^\intercal \bm{\tilde P}_{AB} \underline{u}_{BA}.
	\end{align}
	
	If the true covariance $\bm{\tilde P}$ is known, the \emph{optimal} linear filter is calculable. It is derived from the classical Kalman Filter equations \cite{anderson1979optimal}:
	\begin{subequations}\label{eq: Optimal filter with given correlation}
		\begin{align}
			\underline{\pmb{\hat x}}_F^* &= \underline{\pmb{\hat x}}_A + \frac{\bm{z} - \underline{u}_{BA}^\intercal\left(\underline{\pmb{\hat x}}_A - \underline{\pmb{\hat x}}_B\right)}{\tilde\sigma_A^2+ \tilde\sigma_B^2 - 2\tilde\gamma + \sigma_m^2}(\bm{\tilde P}_A - \bm{\tilde P}_{AB})\underline{u}_{BA}, \\
			\bm{\tilde P}_F^* &= \bm{\tilde P}_A - \frac{(\bm{\tilde P}_A - \bm{\tilde P}_{AB})\underline{u}_{BA}\underline{u}_{BA}^\intercal (\bm{\tilde P}_A - \bm{\tilde P}_{AB}^\intercal)}{\tilde \sigma_A^2+ \tilde \sigma_B^2 - 2\tilde\gamma + \sigma_m^2}.
		\end{align}
	\end{subequations}
	In this case, $\bm{\tilde P}^*_F$ is the minimum (in the Loewner ordering) of the $\bm{\tilde P}_F$, and thus is optimal w.r.t. any increasing function $J$. 
	
	However, as $\bm{\tilde P}$ is not known, such a consideration is not possible, and the linear filter must be consistent, {i.e.}, it must satisfy $\bm{P}_F \succeq \bm{\tilde P}_F$ for every possible covariance. We are now in a position to state the main problem.
	\begin{problem}\label{pro: Main problem}
		Find a gain $\underline{w}_F$ and a covariance $\bm{P}_F$ such that the estimator $(\underline{\pmb{\hat x}}_F, \bm{P}_F)$ defined by \eqref{eq: Linear filter} is consistent and optimal w.r.t. $J$. In other words:
		\begin{equation}\tag{$\text{P}_1$}
			\left\{\begin{array}{ccc}
				\minimize\limits_{\underline{w}_F, \bm{P}_F} & J(\bm{P}_F) & \\
				\subject{} & \forall \bm{\tilde P} \in \boldsymbol{\mathcal{P}}, & \bm{P}_F \succeq \bm{\tilde P}_F
			\end{array}\right.
		\end{equation}
		where $\bm{\tilde P}_F$ is  given by \eqref{eq: MSE}.
	\end{problem}
	
	The main result of this paper is the design of a candidate filter solving Problem~\ref{pro: Main problem} using SCI and its analysis.
	
	\section{SCI for a  distance measurement}\label{sec: SCI for distance measurement}
	
	As a preliminary remark, \eqref{eq: Linear filter} can be rewritten as:
	\begin{subequations}\label{eq: Linear filter (2)}
		\begin{align}
			\underline{\pmb{\hat x}}_F &=  \underline{\pmb{\hat x}}_A + \underline{w}_F\left[\bm{z} - \underline{u}_{BA}^\intercal\left(\underline{\pmb{\hat x}}_A - \underline{\pmb{\hat x}}_B\right)\right]\label{eq: Linear filter as Kalman}\\
			&= \left(\bm{I} - \underline{w}_F \underline{u}_{BA}^\intercal\right) \underline{\pmb{\hat x}}_A + \underline{w}_F \underline{u}_{BA}^\intercal \left(\underline{\pmb{\hat x}}_B + \bm{z}\underline{u}_{BA}\right)\label{eq: Linear filter as fusion}.
		\end{align}
	\end{subequations}
	
	These two expressions highlight the equivalence between linear filtering and fusion. Equation \eqref{eq: Linear filter as Kalman} represents the usual form of the Kalman correction step: $\underline{w}_F$ is a gain and $\bm{z} - \underline{u}_{BA}^\intercal\left(\underline{\pmb{\hat x}}_A - \underline{\pmb{\hat x}}_B\right)$ is the innovation on the measurement. Equation \eqref{eq: Linear filter as fusion} on the other hand represents a fusion: the term $\underline{\pmb{\hat x}}_B + \bm{z}\underline{u}_{BA}$ is another estimator of the state of $A$. Let $\underline{\pmb{\hat x}}_{A}'$ denote this estimator. Equation \eqref{eq: Linear filter as fusion} corresponds to the fusion of $\underline{\pmb{\hat x}}_A$ and the observation of $\underline{\pmb{\hat x}}_{A}'$ through $\underline{u}_{BA}$.
	
	Since the errors $\underline{\pmb{\tilde x}}_{A}$ and $\underline{\pmb{\tilde x}}_{B}$ are correlated to an unknown degree and $\bm{\tilde z}$ is independent of $\underline{\pmb{\tilde x}}_{A}$, the fusion of $\underline{\pmb{\hat x}}_{A}$ and $\underline{\pmb{\hat x}}_{A}'$ respects the assumptions of SCI \cite{julier2001general}.
	SCI fuses two estimators using a linear combination whose weights depend on their covariances. For any $\omega\in[0,1)$, SCI provides a consistent estimator $\left(\underline{\pmb{\hat x}}_\text{SCI}(\omega), \bm{P}_\text{SCI}(\omega)\right)$ for the state of $A$, hereafter called an SCI filter.
	It is obtained by the SCI equations \cite[Eq. (12.24)-(12.25)]{julier2001general}. The covariance is:
	\begin{subequations}\label{eq: Optimal filter both}
		\begin{align}
			\bm{P}_\text{SCI}(\omega) & = \left[(1-\omega)\bm{P}_A^{-1} + \omega\frac{\underline{u}_{BA}\underline{u}_{BA}^\intercal}{\sigma_B^2 + \omega\sigma_m^2}\right]^{-1}\notag\\
			&= \frac{1}{1 - \omega}\left[\bm{P}_A - \frac{\omega \bm{P}_A \underline{u}_{BA}\underline{u}_{BA}^\intercal \bm{P}_A}{\omega \sigma_A^2 + (1-\omega)\left(\sigma_B^2 + \omega\sigma_m^2\right)} \right] \label{eq: Optimal filter covariance},
		\end{align}
		where $\sigma_A^2 = \underline{u}_{BA}^\intercal \bm{P}_A \underline{u}_{BA}$ and $\sigma_B^2 = \underline{u}_{BA}^\intercal \bm{P}_B \underline{u}_{BA}$.
		The corresponding estimator is $\underline{\pmb{\hat x}}_\text{SCI}(\omega) = \underline{\pmb{\hat x}}_F$ given by \eqref{eq: Linear filter (2)} with a gain $\underline{w}_F =\underline{w}_\text{SCI}(\omega)$:
		\begin{equation}
			\underline{w}_\text{SCI}(\omega) = \frac{\omega}{\omega\sigma_A^2 + (1-\omega)\left(\sigma_B^2 + \omega\sigma_m^2\right)}\bm{P}_A \underline{u}_{BA}. \label{eq: Optimal filter gain}
		\end{equation}
	\end{subequations}
	
	The optimal filter w.r.t. $J$ within the family of SCI filter is called the \emph{optimal SCI filter}. To find the optimal SCI filter, the parameter $\omega$ is chosen to minimize the cost function $J$:
	\begin{equation}\label{eq: Omega optimal}
		\omega^* = \underset{0 \le \omega \le 1}{\arg\min{}} J(\bm{P}_\text{SCI}(\omega)).
	\end{equation}
	We have proved the following result.
	\begin{theorem}\label{the: Optimal filter}
		The optimal SCI filter for the state of $A$ is $\left(\underline{\pmb{\hat x}}_\text{SCI}(\omega^*), \bm{P}_\text{SCI}(\omega^*)\right)$ given by \eqref{eq: Optimal filter both} where the parameter $\omega^*$ is defined by \eqref{eq: Omega optimal}.
	\end{theorem}
	
	In the sequel, this optimal SCI filter is denoted $\left(\underline{\pmb{\hat x}}_\text{SCI}^*, \bm{P}_\text{SCI}^*\right)$. SCI has not been shown to be optimal for the fusion, therefore the optimal SCI filter $\left(\underline{\pmb{\hat x}}_\text{SCI}^*, \bm{P}_\text{SCI}^*\right)$ may be suboptimal for Problem~\ref{pro: Main problem}. However, since it is very simple to implement, it is widely used in practice. Furthermore, when the measurement is very accurate, {i.e.}, when $\sigma_m^2$ tends to $0$, SCI becomes CI which is optimal (only when $\sigma_m^2 = 0$)  \cite{reinhardt2015minimum}. These two reasons lead us to consider this filter and analyze its usefulness.

	\section{Usefulness of the filtering}\label{sec: Usefulness of the filtering}
	
	\subsection{General increasing cost function}
	
	Theorem~\ref{the: Optimal filter} gives the expression of the optimal SCI filter. Setting $\omega = 0$ in the fusion ensures that the estimator $\left(\underline{\pmb{\hat x}}_\text{SCI}^*, \bm{P}_\text{SCI}^*\right)$ is at least as good as the original estimator $\left(\underline{\pmb{\hat x}}_A, \bm{P}_A\right)$. In fact, setting $\omega = 0$ corresponds to keeping the original estimator and ignoring the measurement. However, there is no reason why there should be a better estimator, the optimal parameter \eqref{eq: Omega optimal} could be $\omega^* = 0$. In such a case, SCI filters cannot improve the estimator of $A$. We call \emph{pertinent} a linear filter for the state of $A$ that provides a better estimator than $\left(\underline{\pmb{\hat x}}_A, \bm{P}_A\right)$. This section characterizes the pertinence of SCI filters.
	\begin{definition}
		A \emph{pertinent} linear filter for the state of $A$ is a linear filter $\left(\underline{\pmb{\hat x}}_F, \bm{P}_F\right)$ such that $J({\bm{P}_F}) < J({\bm{P}_A})$.
	\end{definition}
	
	By definition of the optimal SCI filter, there is a pertinent SCI filter if and only if the optimal SCI filter is pertinent. The underlying question is therefore: Is the optimal SCI filter $\left(\underline{\pmb{\hat x}}_\text{SCI}^*, \bm{P}_\text{SCI}^*\right)$ pertinent?
	
	Let us start with two corollaries of Theorem~\ref{the: Optimal filter} that provide a necessary condition for the existence of a pertinent SCI filter.
	\begin{corollary}\label{cor: Necessary condtion}
		If $\sigma_A^2 \le \sigma_B^2$, there is not any pertinent SCI filter for the state of $A$.
	\end{corollary}
	\begin{proof}
		Assume that $\sigma_A^2 \le \sigma_B^2$ and let $\omega \in[0,1)$. We will show that $\bm{P}_\text{SCI}(\omega) \succeq \bm{P}_A$ which is sufficient since $J$ is increasing.
		In \eqref{eq: Optimal filter covariance}, $\bm{P}_\text{SCI}^{-1}(\omega)$ is expressed as a convex combination of $\bm{P}_A^{-1}$ and $\bm{Q} = \frac{\underline{u}_{BA} \underline{u}_{BA}^\intercal}{\sigma_B^2 + \omega \sigma_m^2}$. Therefore, if $\bm{Q} \preceq \bm{P}_A^{-1}$, then $\bm{P}_\text{SCI}^{-1}(\omega) \preceq \bm{P}_A^{-1}$. Let us prove that $\bm{Q} \preceq \bm{P}_A^{-1}$, by proving that $\forall \underline{v}$, $\underline{v}^\intercal(\bm{P}_A^{-1}-\bm{Q})\underline{v} \ge 0$. Write $\underline{v} = \alpha \underline{u}_{BA} + \underline{w}$ with $\alpha \in \mathbb{R}$ and $\underline{w}^\intercal\underline{u}_{BA} =0$. Then, $\underline{v}^\intercal\bm{Q}\underline{v} = \alpha^2 / (\sigma_B^2+\omega\sigma_m^2) \le \alpha^2/\sigma_A^{2}$. Finally, let us prove that $\underline{v}^\intercal\bm{P}_A^{-1}\underline{v}\ge \alpha^2/\sigma_A^{2}$.
		
		Consider an orthonormal basis $\mathcal{B} = \left(\underline{u}_{BA}, \underline{u}_2, \dots, \underline{u}_n\right)$ containing $\underline{u}_{BA}$ and the orthogonal matrix $\bm{R} = \begin{bmatrix}\underline{u}_{BA} & \underline{u}_2 & \cdots & \underline{u}_n\end{bmatrix}$. In this basis, the covariance of the error of the estimator of $A$ and the vector $\underline{v}$ become:
		\begin{align*}
			\bm{P}_A &= \bm{R}\begin{bmatrix}
				\sigma_A^2 & \underline{b}^\intercal \\
				\underline{b} & \bm{C}
			\end{bmatrix}\bm{R}^\intercal, &
			\underline{v} &= \bm{R}\begin{pmatrix}\alpha \\ \underline{w}'\end{pmatrix}
		\end{align*}
		where $\underline{b}, \underline{w}' \in \mathbb{R}^{n-1}$ and $\bm{C}\in \mathbb{R}^{(n-1)\times(n-1)}$. Let $\bm{S} = \bm{C} - \underline{b}\underline{b}^\intercal/\sigma_A^2$ be the Schur complement of the first entry. $\bm{S}$ is invertible and $\bm{S}^{-1} \succeq \bm{0}$, see, {e.g.}, \cite[Chap. 7]{horn2012matrix}. The inverse of $\bm{P}_A$ and $\underline{v}^\intercal\bm{P}_A^{-1}\underline{v}$ become:
		\begin{align*}
			\bm{P}_A^{-1} &= \bm{R}\begin{bmatrix}
				\sigma_A^{-2} + \sigma_A^{-4}\underline{b}^\intercal\bm{S}^{-1}\underline{b} & -\sigma_A^{-2}\underline{b}^\intercal\bm{S}^{-1} \\
				-\sigma_A^{-2}\bm{S}^{-1}\underline{b} & \bm{S}^{-1}
			\end{bmatrix}\bm{R}^\intercal,\\
			\underline{v}^\intercal\bm{P}_A^{-1}\underline{v} &= \frac{\alpha^2}{\sigma_A^2} + \left(\frac{\alpha}{\sigma_A^2}\underline{b} - \underline{w}'\right)^\intercal\bm{S}^{-1}\left(\frac{\alpha}{\sigma_A^2}\underline{b} - \underline{w}'\right).
		\end{align*}
		Hence, $\underline{v}^\intercal\bm{P}_A^{-1}\underline{v} \ge \frac{\alpha^2}{\sigma_A^2}$.
	\end{proof}
	\begin{corollary}\label{cor: Only one pertinent}
		If there is a pertinent SCI filter for the state of $A$, then there is not any pertinent SCI filter for the state of $B$.
	\end{corollary}
	\begin{proof}
		If there is a pertinent SCI filter for the state of $A$, Corollary~\ref{cor: Necessary condtion} implies $\sigma_B^2 < \sigma_A^2$. Similarly, if there is a pertinent SCI filter for the state of $B$, then $\sigma_A^2 < \sigma_B^2$. Both inequalities cannot hold simultaneously.
	\end{proof}
	
	Corollary~\ref{cor: Necessary condtion} provides only a necessary condition for the existence of pertinent SCI filters. In the following, we extend the property to provide a necessary and sufficient condition for the two most commonly used cost functions: the trace and the determinant. It is based on the following result.
	\begin{lemma}\label{the: Condition necessary and sufficient}
		Assume that the function $f:\omega \mapsto J(\bm{P}_\text{SCI}(\omega))$, with $\bm{P}_\text{SCI}(\omega)$ given by \eqref{eq: Optimal filter covariance}, is convex on $[0, 1)$. Then, the optimal SCI filter $\left(\underline{\pmb{\hat x}}_\text{SCI}^*, \bm{P}_\text{SCI}^*\right)$ is pertinent if and only if $f'(0) < 0$.
	\end{lemma}
	\begin{proof}
		By definition, $\left(\underline{\pmb{\hat x}}_\text{SCI}^*, \bm{P}_\text{SCI}^*\right)$ is pertinent if and only if $f(\omega^*) < f(0) = J(\bm{P}_A)$. If $f'(0) < 0$, then there is a $\omega_0 > 0$ such that $f(\omega_0) < f(0)$ and therefore $f(\omega^*) < f(0)$. If $f'(0) \ge 0$, since $f$ is convex on $[0,1)$, $\forall \omega \in [0,1)$, $f(\omega) \ge f(0)$ and therefore, $f(\omega^*) = f(0)$.
	\end{proof}
	
	\subsection{Particular case of the trace}
	
	In this paragraph, consider $J(\cdot) = \trace\cdot$ and let $g : \omega \mapsto \trace{\bm{P}_\text{SCI}(\omega)}$ be the cost function to be optimized. According to \eqref{eq: Optimal filter covariance}, the cost function is:
	\begin{equation}\label{eq: cost function trace}
		g(\omega) = \frac{\trace{\bm{P}_A}}{1-\omega}\left[1 - \frac{\omega r_A \sigma_A^2}{\omega \sigma_A^2 + (1-\omega)(\sigma_B^2 + \omega \sigma_m^2)}\right]
	\end{equation}
	where:
	\begin{equation}\label{eq: Ratio}
		r_A = \frac{\frac{1}{\sigma_A^2}\left\lVert\bm{P}_A\underline{u}_{BA}\right\rVert^2}{\trace{\bm{P}_A}}.
	\end{equation}
	
	The meaning of $r_A$ is discussed in Section~\ref{sec: Discussion}. In order to apply Lemma~\ref{the: Condition necessary and sufficient}, let us prove the following.
	\begin{lemma}\label{lem: Convexity of g}
		The cost function $g$ is convex on $[0,1)$.
	\end{lemma}
	\begin{proof}
		Let us first prove that $0 < r_A \le 1$. Consider again the orthonormal basis $\mathcal{B} = \left(\underline{u}_{BA}, \underline{u}_2, \dots, \underline{u}_n\right)$ containing $\underline{u}_{BA}$ and the orthogonal matrix $\bm{R} = \begin{bmatrix}\underline{u}_{BA} & \underline{u}_2 & \cdots & \underline{u}_n\end{bmatrix}$. In this basis, the covariance of the estimator of $A$ becomes:
		\begin{equation}
			\bm{R}^\intercal \bm{P}_A \bm{R} = \begin{bmatrix}
				\sigma_A^2 & \rho_2 \sigma_A\sigma_2 & \cdots & \rho_n \sigma_A\sigma_n \\
				\rho_2 \sigma_A\sigma_2 & \sigma_2^2 & * & * \\
				\vdots & * & \ddots & * & \\
				\rho_n \sigma_A\sigma_n & * & * & \sigma_n^2
			\end{bmatrix}
		\end{equation}
		where only the diagonal coefficients and the correlations with the first component have been labeled. With these notations,  $r_A$ is developed as:
		\begin{equation}\label{eq: ratio 2}
			r_A = \frac{\frac{1}{\sigma_A^2}\left\lVert\bm{P}_A \underline{u}_{BA}\right\rVert^2}{\trace{\bm{P}_A}} 
			= \frac{\sigma_A^2 + \sum_{i=2}^n \rho_i^2\sigma_i^2}{\sigma_A^2 + \sum_{i=2}^n \sigma_i^2}.
		\end{equation}
		Since the correlations satisfy $\left|\rho_i\right| \le 1$, the ratio satisfies $0 < r_A \le 1$.
		
		To prove the convexity of $g$, let us first assume that $r_A <1$. By putting all the terms in the same fraction, the cost function $g$ is a rational function:
		\begin{align*}g(\omega)
			&= \frac{\trace{\bm{P}_A}}{1-\omega}\frac{(1-r_A)\omega \sigma_A^2 + (1-\omega)(\sigma_B^2 + \omega \sigma_m^2)}{\omega \sigma_A^2 + (1-\omega)(\sigma_B^2 + \omega \sigma_m^2)} \\
			&= \frac{\trace{\bm{P}_A}}{1-\omega}\frac{\sigma_B^2 + [(1-r_A)\sigma_A^2-\sigma_B^2 + \sigma_m^2]\omega - \sigma_m^2\omega^2}{\sigma_B^2 + (\sigma_A^2-\sigma_B^2 + \sigma_m^2)\omega - \sigma_m^2\omega^2}.
		\end{align*}
		Let $P:\omega \mapsto \sigma_B^2 + (\sigma_A^2-\sigma_B^2 + \sigma_m^2)\omega - \sigma_m^2\omega^2$ and $Q:\omega \mapsto \sigma_B^2 + [(1-r_A)\sigma_A^2-\sigma_B^2 + \sigma_m^2]\omega - \sigma_m^2\omega^2$ be the polynomials of degree $2$ at the denominator and numerator. Since $P(0) = \sigma_B^2$ and $P(1) = \sigma_A^2$, $P$ has two roots $b$ and $d$ which satisfy: $b < 0 < 1 < d$. Similarly, by noting that $Q(\omega) = P(\omega) - r_A \sigma_A^2 \omega$, $Q(b) = -r_Ab\sigma_A^2 > 0$, $Q(1) = (1-r_A)\sigma_A^2 > 0$ and $Q(d) = -r_A d\sigma_A^2 < 0$, $Q$ has two roots $a$ and $c$ which satisfy:
		\begin{equation*}
			a < b < 0 < 1 < c < d.
		\end{equation*}
		The cost function is therefore:
		\begin{equation*}
			g(\omega) = -\frac{\trace{\bm{P}_A}(\omega - a)(\omega - c)}{(\omega - 1)(\omega - b)(\omega - d)}.
		\end{equation*}
		Then, using partial fraction decomposition:
		\begin{equation}\label{eq: Partial fraction decomposition trace}
			g(\omega) = \trace{\bm{P}_A}\left(\frac{A}{\omega - 1} + \frac{B}{\omega - b} + \frac{C}{\omega - d}\right)
		\end{equation}
		with:
		\begin{align*}
			A &= -\frac{(1-a)(1-c)}{(1-b)(1-d)}<0, & B &= -\frac{(b-a)(b-c)}{(b-1)(b-d)}>0, \\ C &= -\frac{(d-a)(d-c)}{(d-1)(d-b)} < 0.
		\end{align*}
		Since $\trace{\bm{P}_A} > 0$, the three terms in \eqref{eq: Partial fraction decomposition trace} are convex on $(b, 1)$. Therefore, since $b < 0$, the cost function $g$ is convex on $[0,1)$.
		
		Finally, if $r_A = 1$, then $c = 1$ which simplifies the expression of $g$. By the same logic, $g$ is convex on $(b,d)$, and thus on $[0,1]$. Note that $r_A= 1$ corresponds to the one-dimensional case.
	\end{proof}
	
	We are now in a position to apply Lemma~\ref{the: Condition necessary and sufficient} to the trace.
	\begin{theorem}\label{cor: Condition pertinent filter trace}
		The optimal SCI filter for the state of $A$ w.r.t. to the trace is pertinent if and only if:
		\begin{equation}\label{eq: Main result trace}
			\sigma_B^2 < r_A \sigma_A^2.
		\end{equation}
	\end{theorem}
	\begin{proof}
		Since $g$ is convex on $[0, 1)$ as stated in Lemma~\ref{lem: Convexity of g}, by Lemma~\ref{the: Condition necessary and sufficient}, the optimal SCI filter is pertinent if and only if $g'(0) < 0$. The claim then follows from $g'(0) = \trace{\bm{P}_A} \left(1 - r_A {\sigma_A^2}/{\sigma_B^2}\right)$.
	\end{proof}
	
	\subsection{Particular case of the determinant}
	
	In this paragraph, consider $J(\cdot) = \left|\cdot\right|$ and let $h : \omega \mapsto \left|\bm{P}_\text{SCI}(\omega)\right|$ be the cost function to be optimized. According to \eqref{eq: Optimal filter covariance}, the cost function is:
	\begin{equation}\label{eq: cost function determinant (1)}
		h(\omega) = \frac{\left|\bm{P}_A\right|}{(1-\omega)^n}\left|\bm{I} - \frac{\omega \bm{P}_A^{1/2}\underline{u}_{BA}\underline{u}_{BA}^\intercal\bm{P}_A^{1/2}}{\omega \sigma_A^2 + (1-\omega)(\sigma_B^2 + \omega \sigma_m^2)}\right|
	\end{equation}
	where $\bm{P}_A^{1/2}$ denotes the square root of $\bm{P}_A$ and $n$ is the dimension of the state.
	Using that $\left|I + \underline{u}\underline{v}^\intercal\right| = 1 + \underline{u}^\intercal\underline{v}$, \eqref{eq: cost function determinant (1)} becomes:
	\begin{equation}\label{eq: cost function determinant (2)}
		h(\omega) = \frac{\left|\bm{P}_A\right|}{(1-\omega)^{n-1}}\frac{\sigma_B^2 + \omega \sigma_m^2}{\omega \sigma_A^2 + (1-\omega)[\sigma_B^2 + \omega \sigma_m^2]}.
	\end{equation}
	
	As for the trace, the convexity of $h$ is proven in the following result.
	\begin{lemma}\label{lem: Convexity of h}
		The cost function $h$ is convex on $[0,1)$.
	\end{lemma}
	\begin{proof}
		This proof is very similar to the proof of Lemma~\ref{lem: Convexity of g}.
		First, let $P:\omega \mapsto \sigma_B^2 + (\sigma_A^2-\sigma_B^2 + \sigma_m^2)\omega - \sigma_m^2\omega^2$. As $P(0) > 0$, $P(1) > 0$ and $P(-\sigma_B^2/\sigma_m^2) < 0$, the cost function $h$ is a rational function which can be expressed as follows.
		\begin{equation*}
			h(\omega) = -\frac{\left|\bm{P}_A\right|(\omega + \sigma_B^2/\sigma_m^2)}{(1-\omega)^{n-1}(\omega - a)(\omega - b)}
		\end{equation*}
		where the zeros of $P$, denoted $a$ and $b$, satisfy:
		$$-\sigma_B^2/\sigma_m^2 < a < 0 < 1 < b.$$
		Then, using partial fraction decomposition:
		\begin{equation}\label{eq: Partial fraction decomposition determinant}
			h(\omega) = \frac{\left|\bm{P}_A\right|A}{(\omega - a)(1-\omega)^{n-1}} + \frac{\left|\bm{P}_A\right|B}{(\omega - b)(1-\omega)^{n-1}}
		\end{equation}
		with:
		\begin{align*}
			A &= \frac{-(a+\sigma_B^2/\sigma_m^2)}{(a-b)}>0, & B &= \frac{-(b+\sigma_B^2/\sigma_m^2)}{(b-a)}<0.
		\end{align*}
		Both terms in \eqref{eq: Partial fraction decomposition determinant} are convex on $[0,1)$.
	\end{proof}
	
	We are now in a position to apply Lemma~\ref{the: Condition necessary and sufficient} to the determinant.
	\begin{theorem}\label{cor: Condition pertinent filter determinant}
		The optimal SCI filter for the state of $A$ w.r.t. to the determinant is pertinent if and only if:
		\begin{equation}\label{eq: Main result determinant}
			\sigma_B^2 < \frac{1}{n}\sigma_A^2.
		\end{equation}
	\end{theorem}
	\begin{proof}
		Since $h$ is convex on $[0,1)$ as stated in Lemma~\ref{lem: Convexity of h}, the optimal SCI filter is pertinent if and only if $h'(0) < 0$. The claim then follows from $h'(0) = \left|\bm{P}_A\right| \left(n - {\sigma_A^2}/{\sigma_B^2}\right)$.
	\end{proof}
	
	\section{Discussion and numerical simulations}\label{sec: Discussion}
	
	\begin{figure}[t]
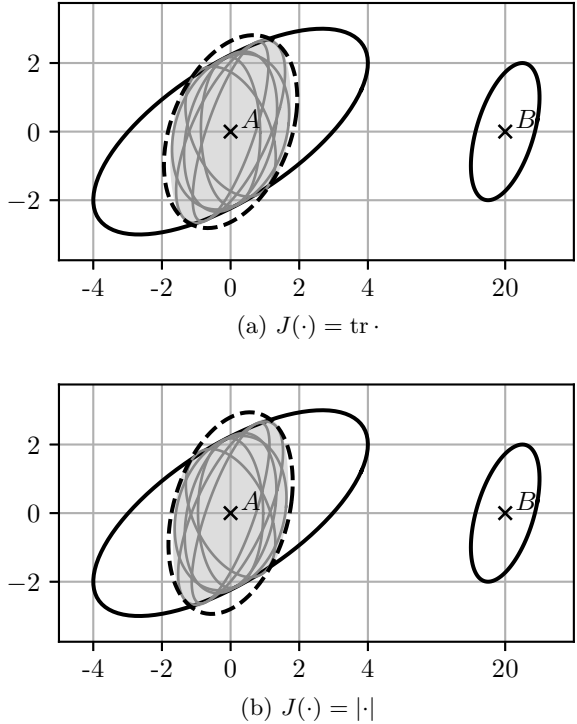

		\begin{subfigure}{\linewidth}
			\centering
			{\input{estimate_with_true_covariances_optimal_trace.tex}}
			\caption{$J(\cdot) = \trace{\cdot}$}
		\end{subfigure}
		\begin{subfigure}{\linewidth}
			\centering
			{\input{estimate_with_true_covariances_optimal_determinant.tex}}
			\caption{$J(\cdot) = \left|\cdot\right|$}
		\end{subfigure}
		\caption{Example of the optimal SCI filters for two different cost functions $J$. The black solid ellipses represent the covariance of the estimators $\bm{P}_A$ and $\bm{P}_B$. The dashed ellipse represents the covariance of the optimal SCI filter $\bm{P}_\text{SCI}^*$. The grey solid ellipses represent possible covariances $\bm{\tilde P}_F$ obtained by \eqref{eq: MSE} with $\underline{w}_F = \underline{w}_\text{SCI}(\omega^*)$ and different $\bm{\tilde P} \in \boldsymbol{\mathcal{P}}$, and the grey area their convex hull. In both figures: $P_A = \left[16, 8; 8, 9\right]$, $P_B = \left[1, 1; 1, 4\right]$, and $\sigma_m^2 = 1$. The optimal parameters are $\omega^* \approx 0.28$ for the trace and $\omega^* \approx 0.36$ for the determinant.}
		\label{fig: Example of optimal filtering}
	\end{figure}
	
	Illustrations of the optimal SCI filters for the trace and the determinant are shown in Fig.~\ref{fig: Example of optimal filtering}. In this figure, the consistency of the filter has been illustrated by generating several possible $\bm{\tilde P} \in \boldsymbol{\mathcal{P}}$ and plotting the resulting MSE given by \eqref{eq: MSE}. As observed, the two optimal SCI filters are different. This means that they depend on the cost function which is a first difference with the usual case where $\bm{\tilde P}$ is known. 
	
	The main implication of Theorem~\ref{the: Optimal filter} and Corollary~\ref{cor: Necessary condtion} is that it is not always possible to improve an estimator using SCI. Improvements require that the \emph{helping} agent, here Agent $B$, has a sufficiently good precision (w.r.t. Agent $A$) in the direction of the measure, {i.e.}, that $\sigma_B^2 < \sigma_A^2$ as stated in Corollary~\ref{cor: Necessary condtion}.
	Furthermore, this condition is only necessary, harder constraints should be satisfied in practice, they depend on the cost function $J$. It can happen that neither agent can improve its estimator, for example when considering the trace, if $r_A\sigma_A^2 < \sigma_B^2 < \sigma_A^2$ as shown in Theorem~\ref{cor: Condition pertinent filter trace}.
	This is an important difference from the usual case. If $\bm{\tilde P}$ is known, it is \emph{almost} always possible to improve the estimator: the optimal linear filter was recalled in \eqref{eq: Optimal filter with given correlation} (Section~\ref{sec: Problem statement}). In this case, even extremely poor precision on the helping agent provides (tiny) improvements.
	
	The expression of the SCI filter \eqref{eq: Optimal filter both} involves the whole statistic of the estimate of Agent $A$ but only the variance of the estimate of Agent $B$ in the direction of the measurement. This asymmetry reduces the communication cost of the filter: the agents only need to send their estimate and their variance in the direction of the measurement (but not their full covariance). The asymmetry is due to the linearization assumption.
	
	Moreover, for the two commonly used cost functions, the trace and the determinant, the necessary and sufficient conditions given in Theorem~\ref{cor: Condition pertinent filter trace} and Theorem~\ref{cor: Condition pertinent filter determinant} allow to check efficiently the pertinence of the SCI filter before doing the measurement. In practice, this criterion can save energy by avoiding useless measurements. Furthermore, the precision of the measurement $\sigma_m^2$ does not appear in these conditions \eqref{eq: Main result trace} or \eqref{eq: Main result determinant}. This means that if the property is not satisfied, even a perfect measurement cannot improve the estimate. However, the variance $\sigma_m^2$ does affect the improvement (if it occurs). As \eqref{eq: Optimal filter both} suggests, the larger $\sigma_m^2$, the smaller the improvement.
	In addition, it seems that the condition for the existence of a pertinent filter becomes more difficult to satisfy as the dimension increases: the right-hand sides of \eqref{eq: Main result determinant} and \eqref{eq: Main result trace} tend to $0$ as the dimension increases (except for the trace when the errors are perfectly correlated in almost all directions).
	
	Finally, it is worth mentioning that the optimal parameter $\omega^*$ can be calculated analytically for the cases of the trace and the determinant. Since the cost functions $g$ and $h$ are convex, the minimum is reached either at $0$, at $1$ (only possible when $r_A = 1$), or when $g'(\omega) = 0$ or $h'(\omega) = 0$. From their decompositions \eqref{eq: Partial fraction decomposition trace} and \eqref{eq: Partial fraction decomposition determinant}, solving $g'(\omega) = 0$ requires finding the roots of a polynomial of degree $4$, and solving $h'(\omega) = 0$ requires finding the roots of a polynomial of degree $3$, both of which can be done analytically.

	\section{Concluding remarks}\label{sec: Conclusion}
	
	The SCI filter is a candidate for the optimal filtering problem. It is based on the linearization of the observation. Under this linearization assumption, we have been able to fully characterize the usefulness of the filter. The validity of this assumption may be evaluated by computing the second-order terms. The extension of this work to nonlinear observation remains open. To better fit the nonlinearities, a possible improvement could be the use of the Unscented Transform \cite{julier2004unscented}.
	Furthermore, the optimal SCI filter proposed here may not be the solution to Problem~\ref{pro: Main problem}: a better linear filter may exist outside the family of SCI filters. We conjecture, after numerous simulations, that the optimal SCI filter is indeed the solution to Problem~\ref{pro: Main problem}. The conjecture is based on the facts that SCI is the natural extension of CI when the errors of the estimators have independent components, and that CI is optimal for the problem it considers \cite{reinhardt2015minimum}.
	
	Finally, we only studied the integration of one distance measurement. Future works will focus on the simultaneous integration of distance measurements with multiple helping agents. They will consider filters such as:
	\begin{align*}
		\underline{\pmb{\hat x}}_F = \bm{K}\underline{\pmb{\hat x}}_A + \bm{L}_1\underline{\pmb{\hat x}}_1 + \bm{z}_1 \underline{w}_1 + \dots  + \bm{L}_m\underline{\pmb{\hat x}}_m + \bm{z}_m \underline{w}_m.
	\end{align*}
	Most of the concepts introduced in this paper can be adapted, but as \cite{ajgl2018fusion} explains, CI is not optimal for more than two estimators, therefore SCI should also be suboptimal.
	
	\bibliographystyle{plain}
	\bibliography{references}
	
\end{document}